\documentclass[12pt,fleqn]{article}

\usepackage[margin=1in]{geometry}
\usepackage{amsmath,amssymb,amsthm,enumitem,pgfplots,setspace,subdepth,tikz,titling,verbatim}
\pgfplotsset{compat=1.18}
\usepackage[round]{natbib}
\usepackage{ctable}

\usepackage{ulem}
\usepackage{accents}

\newcommand{\Pred}{\mathrm{Pred}}
\newcommand{\olsi}[1]{\,\overline{\!{#1}}} %
\newcommand{\wtsi}[1]{\,\widetilde{\!{#1}}} %

\newtheorem{observation}{Observation}
\newtheorem{theorem}{Theorem}
\newtheorem{lemma}{Lemma}
\newtheorem{proposition}{Proposition}
\newtheorem{corollary}{Corollary}
\theoremstyle{definition}
\newtheorem{definition}{Definition}
\newtheorem*{definition*}{Definition}
\newtheorem{axiom}{Axiom}
\newtheorem{assumption}{Assumption}

\newtheorem{remarkinner}{Remark}
\newtheorem*{remark*}{Remark}
\newenvironment{remark}[1][]
  {\begin{remarkinner}[#1]}
  {\hfill\textit{End~of~remark}\end{remarkinner}}

\usepackage{color}

\DeclareRobustCommand{\erase}{\bgroup\markoverwith{\textcolor{red}{\rule[.5ex]{2pt}{0.4pt}}}\ULon}
\newcounter{example}

\author{
    Jens Gudmundsson \quad Jens Leth Hougaard \\ 
    Kohmei Makihara \hspace{2.435em} Alexandros Rigos \\[1em]
    Department of Food and Resource Economics, University of Copenhagen
}
\title{Efficient liability assignment under shock propagation%
\thanks{
    We thank Jay Sethuraman and Lars Peter \O sterdal for valuable comments.
    Financial support from the Independent Research Fund Denmark (grant no. 4260-00050B) is gratefully acknowledged.
    }
}
\begin{document}
  
\maketitle

\begin{abstract}
    \noindent

    We study a model in which shocks propagate along a path chosen by agents embedded in a network. When a shock hits an agent, the affected agent cancels one of her outgoing edges. This cancellation cascades sequentially along a chosen path until reaching a terminal agent, resulting in a systemic cost equal to the sum of individual cancellation losses. A liability rule determines agent payments for realized losses, and we seek to implement efficient path selection in the induced sequential-move game. Our main axiomatic result characterizes a family of rules, which set each agent's liability to be proportional to the system's total realized losses with agent weights depending only on the network structure. We propose a way to set such weights based on a simple path-based procedure that assigns equal importance to all non-sink agents along each path and then aggregates these contributions across paths. These weights coincide with the Shapley value of an associated ``path-counting'' cooperative game and can be computed in polynomial time. A simulation study illustrates the mechanics of our approach.

    \medskip
    \noindent
    \textbf{Keywords:} Network externality, liability assignment, efficient implementation, supply-chain disruptions
\end{abstract}

\section{Introduction}

Network disruptions may have severe economic consequences for affected stakeholders \citep[for an overview, see][]{ElliottGolub} with firm-level shocks propagating into adverse macroeconomic outcomes \citep{Carvalhoetal2021}.
Examples range from disaster-induced production losses \citep{InoueTodo2019} and reallocation externalities in project networks \citep{DhingraEtAl} to cascading bankruptcies through trade credit chains \citep{Jacobson}.
A common theme across these settings is that individual node fragility can escalate into systemic network fragility.
The extent of such escalation depends on factors such as network configuration, diversification strategies, agents' risk-mitigation decisions, and the structure of liability agreements \citep[e.g.][]{BimI,BimII,BirgeEtAl}.
As network efficiency is influenced by decentralized decisions \citep{LeeWhang1999}, a central challenge is to align individual incentives with system-wide objectives. %

To address this challenge, we focus on the role of liability assignment in the context of cascading cancellations triggered by a disruptive event. 
Agents are connected in a directed acyclic network of, say, bilateral contracts in a supply chain.
A shock forces the \emph{source} node to cancel one of its contracts (e.g., a delivery);
the agent affected by this cancellation must then cancel one of her own;
and so on, until the process terminates at a final \emph{sink} node (say, an end consumer).
Each cancellation leads to an economic loss, which may vary with the agents involved.
We do not impose any form of correlation between losses;
for instance, low upstream losses may well be followed by much higher downstream losses.
All losses are summarized in a \emph{loss function}, which maps edges to losses.
We take an axiomatic approach, designing a systematic procedure---a \emph{liability rule}---to allocate the total path losses among the agents in the network. %
Such a rule is especially valuable as a rule-of-thumb for routine operations:
rather than renegotiating liabilities following each disruption, agents commit upfront to a systematic liability scheme that can be applied immediately regardless of the canceled path and realization of losses.

Formally, each liability rule and loss function induces an extensive-form game of perfect information.
Our first and central axiom, \textit{efficient implementation}, asserts that the equilibrium outcomes of this game should coincide with the set of efficient paths.
Second, \textit{realized-loss dependence} centers on a practical feature: 
only on-path losses, which are inherently easier to verify and contract on than off-path counterfactuals, should matter.
Third, \textit{pairwise collusion-proofness} strengthens individual incentive compatibility by ruling out profitable joint deviations.
No agent's unilateral deviation should make any pair strictly better off, as this would enable side payments that undermine efficient implementation.
Finally, we impose a standard regularity condition, \textit{scale invariance}.

Theorem~\ref{TH:fixedweight} shows that these four axioms jointly characterize the family of so-called \textit{fixed-weight} rules.
Such a rule is parameterized by a pre-set weight vector $w$ and assigns to each agent $i$ the fraction $w_i \geq 0$ of the total loss with $\sum_i w_i = 1$.
Weights may depend on the network structure but not on the underlying loss function.
For instance, weights will be positive for agents with multiple outgoing edges (this will be a consequence of \textit{efficient implementation}).
Fixed-weight rules are \emph{reductionist} \citep{BogoMoulin10}: 
liabilities are assigned in the same way as if off-path losses were such that all paths are efficient, at which point there is no reason to condition liabilities on the chosen path.
Yet more, the way losses are distributed along the chosen path is irrelevant:
liabilities depend solely on the total path loss. 
In practice, it thus suffices to verify total losses along the canceled path and there is no need to distinguish some agent $i$'s individual loss from some other agent $j$'s;
especially if losses for some reason are contentious, where it is not objectively clear which cancellation to associate them to, liabilities work the same and we do not need to take a stance.
These rules are very robust and achieve the desired alignment of incentives even in more general settings.%
\footnote{In Section~\ref{SEC:discussion}, we briefly discuss allowing for incomplete information on losses and extending to more general action sets of the agents such as fractional and multi-contract cancellations.}

In this way, a central insight of our analysis is that losses should be treated as a common, systemic responsibility.
Agents must internalize the full downstream consequences of their decisions.
More local liability assignments are inadequate, e.g., in which the agent pays an amount increasing in the ``direct'' loss they cause to their successor (such as the ones introduced in \citealp{GudmundssonHougaardKo2023} and \citealp{HougaardEtAl2017}).
Indeed, such rules may even generate unbounded inefficiency;
see the example in Figure~\ref{FIG:pigou}.%
\footnote{In Appendix~\ref{APP:additional}, we provide a complementary impossibility result. 
Specifically, there are graphs for which no rule satisfies \textit{efficient implementation}, \textit{realized-loss dependence}, and ``on-path only'' liabilities.}
This echoes the observation by \citet{BimI}, who note that ``contracts whose terms are contingent only on individual outputs may be insufficient for dealing effectively with supply risk at higher tiers''. 
Even introducing solidarity among agents on the canceled path (e.g., sharing total losses equally among all on-path agents) generally fails to ensure efficient cancellations.
Fixed-weight rules instead imply a broader notion of solidarity that extends beyond the realized path to off-path agents as well.

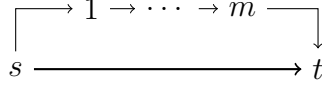
\begin{figure}[!htb]
    \centering
    \begin{tikzpicture}
        \node at (0,0) (s) {$s$};
        \node at (1,.8) (i) {$1$};
        \node at (2,.8) (j) {$\dots{}$};
        \node at (3,.8) (k) {$m$};
        \node at (4,0) (t) {$t$};
        \draw [->] (s) |- (i);
        \draw [->] (i) -- (j);
        \draw [->] (j) -- (k);
        \draw [->] (k) -| (t);
        \draw [thick,->] (s) -- (t);
    \end{tikzpicture}
    \caption{There are two paths from source $s$ to sink~$t$:
        one with many unit-loss edges and one (thick) single-edge path with loss $1+\varepsilon$.
        If the source's liability increases only with the direct loss associated to their choice, then $s$ prefers canceling edge $s \to 1$ over $s \to t$, leading to an inefficient path with total loss $m > 1 + \varepsilon$.}
    \label{FIG:pigou}
\end{figure}

To single out one member of the class of fixed-weight rules, we suggest to set weights reflecting the fraction of source-sink paths that agents are part of.
These weights possess a number of interesting features.
For example, the source is always assigned the highest weight, which is counterbalanced by global solidarity in the sense that all agents (on- and off-path; decision and non-decision makers) have positive weights.
Theorem~\ref{TH:shapley} shows that these weights coincide with the Shapley value of a naturally associated cooperative game in which coalitional worth increases in path count.%
\footnote{The axiomatic foundations for the Shapley value are well understood;
see, for instance, \citet{Shapley1953,Young1985,Neyman1989,vandenBrink2007}.}
Further, the weights can be computed in polynomial time, ensuring that the approach remains tractable in larger instances as well. 

We complement these findings with a simulation study, which highlights some advantages of this efficient rule over the natural status quo (namely, the ``local-liability'' rule in which each agent is liable only for the direct loss their choice causes).
First, due to efficiency, average liabilities are lower and most agents are better off.
Second, as our rule spreads liabilities thin but wide---many agents bear a small share rather than a few bearing the full burden---individual liability variation is significantly lower.
Hence, in a richer framework with risk-averse agents, there are good reasons to believe that our approach may be advantageous to all parties.

\textit{Related literature.}
Our work relates to \citet{GudmundssonHougaardSethuraman25}, who examine a similar setting of propagating network disruptions. 
They restrict to tree graphs in which agents have many incoming edges but only one outgoing, but, more importantly, their focus is different.
Their disruptions follow a pre-specified propagation flow through the network and individual decisions pertain to investments made to avoid disruptions altogether.
Whereas their incentive scheme influences the location and likelihood of disruption, in our work the disruption always occurs at the source and the liability rule instead affects the path selected (and, in turn, the total losses generated).
In this way, our approach bears resemblance also to \citet{GudmundssonHougaardKo2023}, who axiomatically study how to assign liabilities in a ``post-disruption'' setting (covering tree networks in which each node has at most one incoming edge).
The current paper departs from this in that disruptions now have more limited consequences.
Our affected agents need only cancel one outgoing edge, which adds a strategic dimension;
in \citet{GudmundssonHougaardKo2023}, a disruption leads by default to a complete failure of all downstream agreements. 
A related paper is \citet{BakshiKleindorfer2009}, who study a bargaining game with asymmetric information between a retailer and an upstream supplier.
They identify a cost-sharing contract that leads to efficient levels of investments for risk mitigation given certain probabilities of disruption (as represented by a specific parameterized functional form).
This again differs from the approach taken in the current paper, as our model does not feature such investments.

Our paper also relates to the large literature on cost sharing in networks;
see \citet{Hougaard2018} for a recent overview.
Typically, these models feature a central planner who seeks to implement an ``efficient'' network by eliciting information (e.g., on costs or demands) from agents and allocating network costs according to a cost-sharing rule.%
\footnote{In \citet{Moulin14}, the planner elicits connection demands from network users and builds the cheapest network meeting all demands (knowing connection costs).
\citet{AnshelevichEtAl} and \citet{ChenAl} focus on an edge-specific rule that splits costs equally between users.
They examine the induced cost sharing game in which users' strategies consist of paths satisfying their connection demand and study worst-case performance in equilibrium.
\cite{JuarezKumar} consider implementation using cost-sharing rules restricted to depend only on individual path costs and total network cost.}
In comparison, our modeling framework is quite different.
The ``implemented network'' (i.e., the canceled path) is determined by independent agent choices rather than a planner's selection, we do not need to elicit any agent information, and a multiparty contract plays the role of the central planner in adjudicating liabilities.

Having said that, two papers relatively close to ours are \citet{H&T-2015} and \citet{JuarezKoXue}.
In \citet{H&T-2015}, agents have connection demands represented by pairs of nodes on an underlying graph, whereas the planner is ignorant about both connection costs and demands.
Based on user reports, the planner estimates the efficient network and shares the realized costs according to a pre-announced cost-sharing rule.
\citet{H&T-2015} use axioms similar to \textit{realized-loss dependence} and \textit{efficient-path invariance} (see our Proposition~\ref{PR:EPI}) to characterize so-called ``linear simple rules''.
They further examine conditions for which such rules induce truthful reporting in equilibrium.
\citet{JuarezKoXue} consider sequential processes in which agents have individual values at each step (i.e., each edge creates a value for each agent).
A planner selects what paths to implement and how to redistribute the accumulated aggregate value across agents.
They characterize a class of rules (reminiscent of our fixed-weight rules) in a setting in which the planner has complete information. 
The focus of their analysis is on ex-post value sharing for given sequential outcomes, whereas our approach treats the allocation rule as an ex-ante incentive device designed to induce efficient path selection in a decentralized setting.
They also explore an incomplete-information setting, in which the planner only is able to redistribute the realized individual values (agents, on the other hand, still have complete information) and agents vote on which path to implement.
In this case, they primarily single out the ``equal division'' rule.

\textit{Outline.}
In Section~\ref{SEC:model}, we introduce the model, liability rules, and the induced non-cooperative game.
In Section~\ref{SEC:axioms}, we present a number of desirable axioms, leading up to the characterization of the fixed-weight rules.
In Section~\ref{SEC:weights}, we suggest how to set such fixed weights, which turns out to coincide with the Shapley value of an associated cooperative game.
Section~\ref{SEC:simulations} presents the simulation study.
Section~\ref{SEC:discussion} discusses two generalizations of the model.
We conclude in Section~\ref{SEC:concluding}.
Proofs and technical details are postponed to the Appendix.

\section{Model} \label{SEC:model}

In this section, we first introduce the primitives of the model.
There is a fixed network representing agents and their bilateral relations (e.g., joint business projects).
Associated to each project is an economic loss that would be incurred if the project was canceled.
Moreover, cancellations cascade, forming a path through the network.
In the end, losses get reassigned via a \emph{liability rule}, which systematically allocates total cancellation losses to agents.
Such a rule induces a non-cooperative game in which agents make strategic cancellation decisions with the objective of minimizing their own liability.
From a systemic perspective, our aim is to identify liability rules for which the set of equilibria of the induced game coincides with the set of efficient cancellation paths.  

\subsection{Primitives} \label{SUB:primitives}

There is a directed graph $(N,E)$ that describes a network of at least three \textbf{agents} $N$ and their relations $E \subseteq N \times N$ (interchangeably, ``nodes'' and ``edges''), which is held fixed throughout.
We interpret the graph broadly to capture collaborative projects or input-output relations in a supply chain, where an edge $ij \in E$ can be interpreted as agent $i$ delivering inputs to agent~$j$.
The graph is connected, acyclic in the directed sense, and has a unique \emph{source} node~$s$ without incoming edges.%
\footnote{It is straightforward to extend to a network with multiple nodes without incoming edges as long as we maintain the assumption that only one is directly affected by the shock; 
compare Section~\ref{SEC:simulations}.
Agents who cannot be affected even indirectly (i.e., nodes not reachable from the shock) are then considered free of liability.}
Let $i \to j$ whenever there is an edge $ij \in E$.
As the graph is acyclic, there is a topological (total) ordering $\leq$ of~$N$ such that $i \to j \implies i < j$ and there is at least one \emph{sink} node without outgoing edges (e.g., representing end consumers).
A \textbf{path} is a sequence of adjacent agents (or, interchangeably, adjacent edges) connecting the source to a sink.
Let~$\mathcal{P}$ be the set of paths.
Moreover, let $N_P \subseteq N$ and $E_P \subseteq E$ denote the nodes and edges, respectively, of path~$P\in {\cal P}$. 
Let $\mathcal{P}^-_i$ be the set of subpaths (``histories'') from the source to agent~$i$ and $\mathcal{P}^+_i$ be the subpaths (``continuations'') from agent~$i$ to a sink.

We impose a mild richness condition, namely that no interior agent can be part of all paths in~$\mathcal{P}$.
An immediate implication is that the source has multiple outgoing edges. 

\begin{assumption}[No bottlenecks] \label{A2}
    For each non-source and non-sink agent~$i$, there is a path $P \in \mathcal{P}$ with $i \not \in N_P$.
\end{assumption}

The situation we consider is that the source~$s$ is hit by an exogenous shock.
In consequence, the source has to select one of its outgoing edges to cancel;
say $s$ cancels $(s \to i_1)$.
This has further repercussions:
now agent~$i_1$ also has to cancel one of its edges, say~$(i_1 \to i_2)$, and so forth.
The cancellations continue to cascade to form a path connecting the source to a sink.
Each edge $e \in E$ has an associated \textbf{loss} $\ell(e) \geq 0$ that is incurred if the edge is canceled.
We think of the loss $\ell(ij)$ as the costs incurred by agent $j$ if agent $i$ fails to deliver its inputs, but other interpretations are also possible.
Let $\ell \colon E \to \mathbb{R}_{\geq 0}$ describe the \textbf{loss function} and collect all such functions in $\mathcal{L} \equiv \{ \ell \colon E \to \mathbb{R}_{\geq 0} \}$.
For a path or subpath~$P$, define $\ell(P) \equiv \sum_{e \in E_P} \ell(e)$ as the total loss of $P$.

\subsection{Efficient paths}

Our primary objective is to incentivize efficient path cancellation.
For this purpose, collect the cheapest, or \textbf{efficient}, paths in
\[ \textstyle
    \mathcal{E}(\ell) \equiv \arg \min_{P\in\mathcal{P}} \ell(P) \subseteq \mathcal{P}.
\]
Denote the cost of the cheapest $i$-to-sink subpath by $L_i \equiv\min_{ P \in\mathcal{P}^+_i} \ell(P) \geq 0$ and the set of such subpaths by
\[ \textstyle
    \mathcal{E}_i(\ell) \equiv \arg \min_{P\in\mathcal{P}^+_i} \ell(P) \subseteq \mathcal{P}^+_i.
\]
The acyclic graph structure ensures that these costs and paths can be computed quickly, for instance using \citeauthor{Dijkstra1959}'s~(\citeyear{Dijkstra1959}) shortest path algorithm. 
    
An immediate observation is that efficient paths are dynamically consistent \citep[e.g.,][]{Bellman1957}.
That is to say, if we decompose an efficient path $P$ into subpaths $Q = (s \to \dots \to i)$ and $R = (i \to \dots \to t)$ around an on-path agent $i \in N_P$, then $Q$ is a cheapest subpath from $s$ to $i$ and $R$ a cheapest subpath from $i$ to a sink~$t$.

\subsection{Liability rules}

In terms of timing, the network structure $(N,E)$ is fixed from the outset.
Agents anticipate that disruptions may occur, but the nature and severity of the shock, and hence the induced loss function, may vary.
Rather than renegotiating liabilities each time a disruption occurs, agents contract in advance on a \emph{liability rule}, which systematically reassigns the incurred losses as a function of the canceled path~$P$ and the loss function~$\ell$.
Such ex-ante contracting could for instance be implemented through blockchain-based smart contracts \citep[compare][]{GudmundssonHougaardKo2023,GudmundssonHougaard2026}.
Such a contract can automate liability assignment and reduce transaction costs, for instance those arising from legal disputes.%
\footnote{The case in which the liability rule also is contingent on the location of the shock is discussed in Section~\ref{SEC:simulations}.}

\begin{definition}
    A liability rule $\phi \colon \mathcal{P} \times \mathcal{L} \to \mathbb{R}^N_{\geq 0}$, a \textbf{rule} for short, assigns liability $\phi_i(P, \ell) \geq 0$ to agent~$i$ if path $P$ is canceled under loss function~$\ell$ such that $\sum_i \phi_i(P,\ell) = \ell(P)$.
\end{definition}

Restricting to non-negative liabilities is a mild solidarity assumption:
no agent should profit from disruptions if others are harmed.
Balance is desirable to eliminate any need to inefficiently ``burn'' resources or to rely on subsidies from external parties.
Once the shock occurs, the corresponding loss function is common knowledge.
Hence, once agents choose edges to cancel, they do so knowing the graph, loss function, and liability rule. 
Although efficient paths thus are computable, they will only be realized if the liability rule correctly aligns individual incentives.
Our objective is to design rules in such a way that the canceled path, jointly selected through independent choices, is efficient no matter the loss function.
To formalize this, we next cast the interaction as a non-cooperative game.

\subsection{The induced game and efficient implementation}

Fix a rule~$\phi$ and loss function $\ell \in \mathcal{L}$. 
The \emph{game} induced by~$\phi$ and parameterized by~$\ell$ is denoted $\mathcal{G}(\ell; \phi)$.
This is an extensive-form game with perfect information. 
Its \emph{players} are the agents~$N$.
The source, agent~$s$, moves first and chooses its successor~$i$ among its neighbors $N^+_s = \{ i \in N \mid s \to i \}$.
Thereafter, $i$ chooses its successor~$j$ from $N^+_i = \{ j \in N \mid i \to j \}$, and so on.
A \emph{strategy} for an agent~$i$ is a function~$\sigma_i \colon \mathcal{P}^-_i \to N^+_i$, which assigns a successor $\sigma_i(P) \in N^+_i$ to each history $P \in \mathcal{P}^-_i$ leading up to~$i$.
Let~$\sigma$ denote the \emph{strategy profile}.
Once we reach a sink, the game ends.
Let $\Pi(\sigma) \in \mathcal{P}$ denote the source-sink path realized under strategy profile~$\sigma$.
Each agent~$j$ chooses strategy~$\sigma_j$ to minimize its liability $\phi_j(\Pi(\sigma_j,\sigma_{-j}), \ell)$ with common knowledge on $\phi$ and~$\ell$.
We take subgame-perfect equilibrium as solution concept \citep[e.g.][]{OsborneRubinstein1994}.
Let $\mathcal{S}(\ell; \phi) \subseteq \mathcal{P}$ denote the set of paths that are subgame-perfect equilibrium outcomes of the game $\mathcal{G}(\ell; \phi)$.

We wish to design the rule such that the set of paths realized under equilibrium play of the induced game coincide with the efficient paths \citep[``exact'' or ``full'' implementation, see e.g.][]{DasguptaHammondMaskin1979}.

\begin{axiom}[Efficient implementation]
    For every loss function $\ell \in \mathcal{L}$,
    \[
        \mathcal{S}(\ell; \phi) = \mathcal{E}(\ell).
    \]
\end{axiom}

There are many rules that satisfy \textit{efficient implementation}.
One example is \emph{equal division}:
share total losses equally among all agents.
Indeed, this is a member of the class of rules that we will characterize and possesses a form of solidarity property in that losses are borne by all agents, not only the ones directly affected by the cancellations.

Having said that, \textit{efficient implementation} can also be satisfied if losses are shared only along the canceled path.
For instance, one can assign the full efficient cost $L_s$ to the source and, for each on-path agent~$i$, assign liability equal to the difference between realized and efficient cost from~$i$ to the sink (that is, agents bear the systemic inefficiency associated to their choices).
Although such a rule achieves efficiency, it relies on counterfactual losses and liabilities are sensitive to losses off the realized path.
In what follows, we will argue against such rules due to practical concerns pertaining to enforceability, compliance, and stability.
That is to say, it may be far easier to contract on realized losses than on counterfactual ones.

\section{Axiomatic foundations} \label{SEC:axioms}

We introduce three additional desirable properties of rules:
\textit{realized-loss dependence}, \textit{pairwise collusion-proofness}, and \textit{scale invariance}. 
Our main axiomatic result, Theorem~\ref{TH:fixedweight}, shows that, together with \textit{efficient implementation}, these properties characterize a class of rules that we term \emph{fixed-weight rules}.

\subsection{Realized-loss dependence}

Recall that we interpret a liability rule as a contractual agreement between the agents. 
For this to be enforceable in practice, ideally the rule should depend only on verifiable outcomes;
that is to say, whereas realized losses are observable ex post and can serve as the basis of enforceable transfers, counterfactual losses are inherently unverifiable and difficult to contract upon.
In this spirit, our next axiom prohibits rules that depend on off-path losses.
Such rules include, for instance, ones that hold agents liable for ``the externalities they impose on others'' by comparing realized losses with some counterfactual losses.%
\footnote{This property resembles the axiom ``Unobserved information independence'' in \cite{H&T-2015}.
Somewhat similar ideas have been considered in \cite{ChenAl} and \cite{JuarezKumar} creating minimum information frameworks for implementation of efficient connection networks.}

We stress that this restriction on liability rules does not require abandoning efficiency.
As we shall see, the axioms are compatible. 
Rather, \textit{realized-loss dependence} plays the role among our axioms of ensuring that efficiency is achieved in a stable and practically feasible manner.
We do not need complicated contracts based on counterfactual losses, it suffices to restrict attention to incentive schemes grounded in realized and verifiable losses.

\begin{axiom}[Realized-loss dependence]
    For each path $P \in \mathcal{P}$ and all loss functions $\ell, \ell' \in \mathcal{L}$,
    \[ 
        \ell(e) = \ell'(e) \text{ for each edge } e \in E_P 
        \, \implies \,
        \phi(P, \ell) = \phi(P, \ell'). 
    \]
\end{axiom}

Besides ensuring a practically viable rule, \textit{realized-loss dependence} also adds a level of robustness to \textit{efficient implementation}.
First, an agent's liability must increase in the total losses of the cheapest continuation following the agent's choice, even off the equilibrium path.
(As no inefficient subpaths are ever realized, liabilities of such paths do not need to be restricted.)
We record this in Proposition~\ref{PR:eff}, where $\oplus$ is used to denote the operation of concatenating subpaths and edges.

\begin{proposition}[Downstream monotonicity] \label{PR:eff}
    Let $\phi$ satisfy \textit{realized-loss dependence} and \textit{efficient implementation}.
    Consider an arbitrary loss function $\ell\in\mathcal{L}$, agent $i \in N$, and subpath $P \in \mathcal{P}^-_i$.
    For all agents $j, k \in N^+_i$ and efficient continuation subpaths $P_j\in \mathcal{E}_j(\ell)$ and $P_k\in \mathcal{E}_k(\ell)$, 
    \[
        \ell(ij) + L_j < \ell(ik) + L_k 
        \, \iff \,
        \phi_{i}(P \oplus ij \oplus P_j,\ell) < \phi_{i}(P \oplus ik \oplus P_k,\ell).
    \]
\end{proposition}

An immediate implication of Proposition~\ref{PR:eff} is the following.
Fix a rule $\phi$ satisfying the two properties together with an arbitrary loss function $\ell \in \mathcal{L}$ and equilibrium strategy profile $\sigma \in \mathcal{S}(\ell; \phi)$.
Then, for each agent $i$ and history $P \in \mathcal{P}^-_i$,
\[ \textstyle
    \sigma_i (P) \in \arg \min_{j \in N^+_i} \left ( \ell(ij) + L_j \right ).
\]
That is to say, not only do agents make efficient decisions in equilibrium, but the same extends even off the equilibrium path.
We denote this property \textit{robust efficient implementation}. 
We remark also that \textit{robust efficient implementation} is a genuine strengthening of \textit{efficient implementation}:
there are rules that satisfy the latter but not the former.%
\footnote{For instance, let the rule $\phi$ ``punish'' the first agent to make an inefficient choice.
If $P \in \mathcal{E}(\ell)$, then share losses equally: $\phi_i(P,\ell) = 1/n \cdot \ell(P)$;
otherwise, for $P' \not \in \mathcal{E}(\ell)$, there is a first agent $i$ to make an inefficient choice (that is, all paths resulting from $i$'s choice are inefficient).
In this case, assign full liability to~$i$:
$\phi_i(P',\ell) = \ell(P')$.
This rule satisfies \textit{efficient implementation}, but off the equilibrium path agents are free of liability and efficient continuations are not guaranteed.
Hence, it fails \textit{robust efficient implementation}.} 

\begin{corollary} \label{CO:eff}
    \textit{Efficient implementation} and \textit{realized-loss dependence} jointly imply \textit{robust efficient implementation}.
\end{corollary}

\subsection{Pairwise collusion-proofness}

Whereas \textit{efficient implementation} addresses individual incentives, our next axiom pertains to group collusion.
Take an equilibrium path $P$ and a pair of agents~$i$ and~$j$.
It should then not be the case that a strategic deviation by $i$ makes the pair jointly improve.
If such a deviation was beneficial, then %
$i$ and~$j$ could arrange side payments to make them both better off.%

\begin{axiom}[Pairwise collusion-proofness]
    Consider a loss function $\ell \in \mathcal{L}$ and equilibrium strategy profile $\sigma$ with path $P = \Pi(\sigma) \in \mathcal{S}(\ell; \phi)$. 
    For each agent $i \in N$, strategy $\sigma'_i \colon \mathcal{P}^-_i \to N^+_i$ inducing path $P' = \Pi((\sigma'_i,\sigma_{-i}))$, and agent $j \in N$,
    \[
        \phi_i(P,\ell) + \phi_j(P,\ell) \leq \phi_i(P',\ell) + \phi_j(P',\ell).
    \]
\end{axiom}

We next find that, when taken together, \textit{efficient implementation} and \textit{pairwise collusion-proofness} imply that when multiple efficient paths exist, then it makes no difference which we select.
This is in itself desirable to circumvent any form of potential conflicts between agents on which efficient path to target.
Restated in the language of implementation theory, the problem can be viewed as associating to each loss function a set of outcomes (here the efficient paths);
the axioms then imply that the corresponding social choice correspondence is \textit{essentially single-valued} \citep[e.g.][]{MaskinMoore1999,MaskinSjostrom2002}.

\begin{proposition}[Efficient-path invariance] \label{PR:EPI}
    Let $\phi$ satisfy \textit{efficient implementation} and \textit{pairwise collusion-proofness}.
    For each loss function $\ell \in \mathcal{L}$ and all efficient paths $P,P' \in \mathcal{E}(\ell)$, 
    \[
        \phi(P,\ell) = \phi(P',\ell).
    \]
\end{proposition}

The proof is straightforward.
Take two distinct efficient paths $P$ and $P'$, so $\ell(P) = \ell(P')$, and let $i$ be the first agent to make different choices under the two.
Any strict preference for $i$ between the two paths would rule out one as an equilibrium outcome, contradicting \textit{efficient implementation}.
Moreover, any strict preference for some agent $j \neq i$ between the two would allow $i$ and $j$ to collude, contradicting \textit{pairwise collusion-proofness}.
Hence, liabilities must be the same for the two paths.

Next, Proposition~\ref{PR:redistribution} shows that rules satisfying our basic principles---%
the fundamental efficiency-incentive alignment of \textit{efficient implementation},
the informational constraint of \textit{realized-loss dependence},
and the stability property of \textit{pairwise collusion-proofness}---must be \emph{redistribution invariant}.
That is to say, although liabilities $\phi(P,\cdot)$ can depend on total losses~$\ell(P)$, they cannot depend on where each loss is incurred along the edges of~$P$.

\begin{proposition}[Redistribution invariance] \label{PR:redistribution}
    Let $\phi$ satisfy \textit{efficient implementation}, \textit{realized-loss dependence}, and \textit{pairwise collusion-proofness}.
    For each path $P \in \mathcal{P}$ and all loss functions~${\ell,\ell' \in \mathcal{L}}$, 
    \[
        \ell(P) = \ell'(P)
        \, \implies \,
        \phi(P,\ell) = \phi(P,\ell').
    \]
\end{proposition}

The proof of Proposition~\ref{PR:redistribution} is a bit more involved.
We decompose the redistribution from $\ell$ to $\ell'$ into a sequence $\ell, \ell^1, \ell^2, \dots, \ell'$ of local adjustments that preserve the total loss on the path.
In each step, the respective loss function is designed to keep the relevant paths efficient and leave realized losses unchanged for suitably chosen comparison paths.
\textit{Efficient-path invariance} (Proposition~\ref{PR:EPI}) ensures that liabilities are equal across paths at each step, whereas \textit{realized-loss dependence} carries the argument across steps.
By iterating this reasoning, we find that redistributing losses along a path will not affect liabilities.

Next, Proposition~\ref{PR:pathindependence} offers a complementary finding:
the rules satisfying our desirable axioms must also be \textit{path independent}.
That is to say, liabilities $\phi(\cdot,\ell)$ cannot depend on which of the equally-costly paths was chosen. 

\begin{proposition}[Path independence] \label{PR:pathindependence}
    Let~$\phi$ satisfy \textit{efficient implementation}, \textit{realized-loss dependence}, and \textit{pairwise collusion-proofness}.
    For all paths~$P,P' \in \mathcal{P}$ and each loss function~${\ell \in \mathcal{L}}$, 
    \[
        \ell(P) = \ell(P') 
        \, \implies \,
        \phi(P,\ell) = \phi(P',\ell).
    \]
\end{proposition}

To sketch the argument, we go via an alternative loss function $\hat\ell$ that leaves losses on~$P$ unchanged while rendering all paths efficient.
Under $\hat\ell$, both $P$ and $P'$ are efficient and, by \textit{efficient-path invariance} (Proposition~\ref{PR:EPI}), liabilities are equal.
Since $\hat\ell$ coincides with~$\ell$ on $P$, \textit{realized-loss dependence} implies that liabilities on $P$ are unchanged.
Moreover, as~$P$ and~$P'$ have the same total loss under both loss functions, $\hat\ell$ differs from $\ell$ along $P'$ only by a redistribution;
by Proposition~\ref{PR:redistribution}, liabilities are again equal.
Taken together, this yields $\phi(P, \ell) = \phi(P', \ell)$.

An immediate consequence of Propositions~\ref{PR:redistribution} and~\ref{PR:pathindependence} is that liabilities should depend only on the path's total losses.

\begin{corollary} \label{CO:total}
    Let~$\phi$ satisfy \textit{efficient implementation}, \textit{realized-loss dependence}, and \textit{pairwise collusion-proofness}.
    For all paths $P,P' \in \mathcal{P}$ and loss functions~$\ell,\ell' \in \mathcal{L}$, 
    \[
        \ell(P) = \ell'(P') 
        \, \implies \,
        \phi(P,\ell) = \phi(P',\ell').
    \]
\end{corollary}

\subsection{Scale invariance and the class of fixed-weight rules}

The final axiom is the conventional assumption that large and small problems are solved alike;
that is, the rule is \textit{scale invariant}.
Scaling all losses by a common factor scales liabilities accordingly.

\begin{axiom}[Scale invariance]
    For each path $P \in \mathcal{P}$, loss function $\ell \in \mathcal{L}$, and scalar $\alpha > 0$,
    \[
        \phi(P, \alpha \cdot \ell) = \alpha \cdot \phi(P, \ell).
    \]
\end{axiom}

As noted in Corollary~\ref{CO:total}, \textit{efficient implementation}, \textit{realized-loss dependence}, and \textit{pairwise collusion-proofness} imply that liabilities should depend only on the path's total losses.
With the additional requirements of \textit{scale invariance}, we can pin down a parameterized family of rules.
For this purpose, let first $\Delta^N = \{ w \in \mathbb{R}^N_{\geq 0} \mid \sum_i w_i = 1\}$ denote the usual $N$-simplex.
We associate to each point $w \in \Delta^N$ a rule that always shares total losses in proportion to the weights~$w$.

\begin{definition}[Fixed-weight rule $\phi^w$ with parameter~$w \in \Delta^N$]
    For each path $P \in \mathcal{P}$ and loss function $\ell \in \mathcal{L}$,
    \[
        \phi^w(P,\ell) = w \cdot \ell(P).
    \]
\end{definition}

In a sense, a fixed-weight rule simplifies the design problem from determining how losses should be assigned in each contingency to determining how total losses should be weighted across agents as a function of the network structure. 
That is, weights can be conditioned on graph-based primitives and normative considerations;
thereafter, liabilities scale accordingly.
Standard examples include \emph{equal division} ($w_i = w_j$ for all agents $i$ and $j$), weights proportional to the number of paths that agents are on or to source distance, as well as any form of node centrality measure \citep[e.g.][]{Hougaard2018,Jackson}.

It is not difficult to see that fixed-weight rules satisfy \textit{realized-loss dependence}, \textit{pairwise collusion-proofness}, and \textit{scale invariance}.
To ensure \textit{efficient implementation}, we need to impose one additional restriction on weights.
As agents with multiple outgoing edges need to be incentivized to make efficient choices (and there are loss functions under which it is critical that they do), 
they must be assigned positive weight.%
\footnote{To see why, take as example $w \in \Delta^N$ with $w_s = 0$:
that is to say, suppose that the source is free of liability.
Hence, we always have $\phi_s(P,\ell) = 0$, so the source is indifferent between all paths---including inefficient ones.
Then there can be inefficient equilibria, contradicting \textit{efficient implementation}.
To extend the conclusion to any agent $i$ with multiple outgoing edges, we apply the argument for the particular case in which the loss function $\ell$ is such that $i$ is part of at least one efficient and one inefficient path.}
Therefore, let 
\[
    \Delta^N_\star = \{ w \in \Delta^N \mid w_i > 0 \text{ for agents $i \in N$ with $\lvert N^+_i \rvert > 1$} \}. 
\]
We are now ready to state our main axiomatic result.
Independence of the axioms is shown in Appendix~\ref{APP:independence}.

\begin{theorem}[Characterization of fixed-weight rules] \label{TH:fixedweight}
    A rule~$\phi$ satisfies \textit{efficient implementation}, \textit{realized-loss dependence}, \textit{pairwise collusion-proofness}, and \textit{scale invariance} if and only if $\phi = \phi^w$ for some $w \in \Delta^N_\star$.
\end{theorem}

The result of Theorem~\ref{TH:fixedweight} highlights a surprising implication: 
agents who \emph{can} make choices must bear a positive share of \emph{every} realized loss---even those arising on paths that completely bypass the agent and for which the agent's choices are irrelevant.
At its extreme, consider Figure~\ref{FIG:solidarity}, in which dashed edges have zero losses and solid edges have positive losses.
Even though total losses stem from the source's choice whereas agent $j$'s choice adds no additional harm, we will still have that $j$ is assigned positive liability;
for instance, we will have $\phi_j(P,\ell) > 0$ for path $P = (s \to i \to t)$.
Put differently, agent $j$ bears a higher cost than the total losses of the subgraph from $j$ and on, $\ell(jk) + \ell(jt) = 0$.

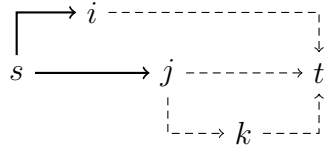
\begin{figure}[!htb]
    \centering
    \begin{tikzpicture}
        \node at (0,0) (s) {$s$};
        \node at (1,.8) (i) {$i$};
        \node at (2,0) (j) {$j$};
        \node at (3,-.8) (k) {$k$};
        \node at (4,0) (t) {$t$};
        \draw [thick,->] (s) |- (i);
        \draw [thick,->] (s) -- (j);
        \draw [densely dashed,->] (j) |- (k);
        \draw [densely dashed,->] (i) -| (t);
        \draw [densely dashed,->] (j) -- (t);
        \draw [densely dashed,->] (k) -| (t);
    \end{tikzpicture}
    \caption{Solid edges have positive losses, dashed edges have zero losses.} 
    \label{FIG:solidarity}
\end{figure}

Fixed-weight rules not only incentivize efficient implementation at the individual level, but even at the group level.
That is to say, if a group $S \subseteq N$ attempted to ``collude'' and coordinate their decisions, then they would still jointly pay $\sum_{i \in S} \phi_i(P,\ell) = \sum_{i \in S} w_i \cdot \ell(P)$.
Hence, they can do no better than minimizing $\ell(P)$ by choosing an efficient path. 
In this way, \textit{pairwise collusion-proofness} joint with our other axioms imply a stronger form of ``full'' collusion-proofness.
(This is reminiscent of properties such as group strategy-proofness, see e.g. \citealp{Barberaetal2016}.)
Further robustness arguments in favor of fixed-weight rules are discussed in Section~\ref{SEC:discussion}.

\begin{remark}[Necessity of ``no bottleneck agents'' for fixed-weight characterization] \label{REM:bottlenecks}
    If we do not impose Assumption~\ref{A2}, then there are rules outside the fixed-weight family that satisfy our axioms.
    For example, take a line graph such as $s \to i \to t$.
    As there is only one path, axioms such as \textit{efficient implementation} and \textit{realized-loss dependence} are vacuous (the unique path is always efficient, and all losses are always realized).
    Hence, we could for instance let each agent pay the loss they ``generate'';
    let $\widehat\phi$ be such that $\widehat\phi_s(P,\ell) = \ell(si)$ and $\widehat\phi_i(P,\ell) = \ell(it)$.
    This is \textit{scale invariant} yet falls outside the fixed-weight family.
    \end{remark}

\section{Choosing weights} \label{SEC:weights}

In this section, we single out a particular method of setting weights.
The underlying principle is that agents appearing on many potential cancellation paths, especially shorter ones, should receive higher weight.
Hence, a natural approach is to proceed path by path and assign weights reflecting the frequency and importance of an agent's presence across paths.
This parallels ideas in cooperative game theory that more ``pivotal'' agents receive higher shares, for example as in the \citeauthor{Shapley1953} (\citeyear{Shapley1953}) value.
Indeed, we will find that the weights~$w^*$ to be introduced next coincide with the Shapley value of a naturally associated cooperative game.
In contrast to the general case in which the value may be computationally intractable \citep{Chalkiadakisetal2012}, our model allows a convenient shortcut to quickly compute such Shapley-like weights $w^*$ even for large problems.

We assume that sinks $t$ are assigned zero weight.
This mirrors the logic of \textit{strict liability} (that an injurer is responsible for the losses she causes) from the literature on law and economics \citep[e.g.][]{Shavell2007}.
Here, in particular, sinks incur losses but cause none, so we take as given that they are free of liability \citep[compare also][]{GudmundssonHougaardKo2023}.
With this in mind, it will turn out to be convenient to work with the set of non-sink agents.
We denote them $N^* \equiv \{ i \in N \mid N^+_i \neq \emptyset \} \subset N$ and, analogously, let $N^*_P \equiv N_P \cap N^*$ be the non-sink agents of path~$P$.

\subsection{Path-by-path algorithm} \label{SUB:weightalgo}

We introduce first a simple algorithm to compute the weights $w^* \in \Delta^N_{\star}$. 
We proceed path by path.
Let $\lvert \mathcal{P} \rvert$ denote the number of paths.
All agents on a path are treated the same;
a value of $1 / \lvert\mathcal{P}\rvert$ is shared between all non-sink agents on the path.
Thereafter, we sum these values over all paths. 
As each agent is part of at least one path, all non-sink weights are positive.

\begin{definition}[Weights $w^* \in \Delta^N_{\star}$]
    For each path $P \in \mathcal{P}$ and agent~$i \in N$, define 
    \[
        v_i(P) = 
        \begin{cases} 
            1 / (\lvert N^*_P \rvert \cdot \lvert \mathcal{P} \rvert) & \text{if $i \in N^*_P$} \\
            0 & \text{otherwise}
        \end{cases}
    \]
    and set weight $w^*_i = \sum_{P \in \mathcal{P}} v_i(P)$. 
\end{definition}

Applied to the case of a tiered supply-chain network (with nodes representing suppliers, manufacturers, retailers etc. and sinks corresponding to end consumers), these weights may admit an even simpler interpretation.
Specifically, assume that edges link each layer only to the next one, and does so in a ``symmetric'' way:
all nodes in a given layer have the same number of outgoing edges and the same number of incoming edges.
In this case, $w^*$ is computed by double application of equal division:
first total losses are shared equally across layers, and then equally among nodes within each layer.%
\footnote{In this case, all paths are of the same length and include one node from each layer.
As agents within a layer have the same in- and out-degrees, they are part of the same number of paths.}
Hence, liabilities are highest in ``thin'' layers (e.g., bottleneck-like central distribution nodes).

\subsection{Relation to the Shapley value} \label{SUB:Shapley}

We continue our analysis of the fixed-weight rule with weights $w^*$ by highlighting a connection to cooperative game theory \citep[see e.g.,][]{PelegSudholter2007}.
Define the \textbf{path-counting game} $v$ with players $N^*$ (i.e., non-sink agents) as follows.
For every coalition $S \subseteq N^*$, let the worth $v(S) \in [0,1]$ be the fraction of all paths that only cover nodes in~$S$:
\[
    v(S) = \frac{\lvert\{ P \in \mathcal{P} \mid N^*_P \subseteq S \} \rvert}{\lvert\mathcal{P}\rvert}.
\]
For instance, if $s \to i \to t$ denotes a path, then this is counted for all coalitions $S$ with $\{ s, i \} \subseteq S$. 
As all paths include the source, coalitions $S$ that do not include the source $s$ have zero values.
For the (grand) coalition of non-sink agents, we have $v(N^*) = 1$.

Theorem~\ref{TH:shapley} shows that the weights $w^*$ coincide with the Shapley value of~$v$.
To build intuition, associate to each path $P \in \mathcal{P}$ the \emph{simple} game $v^P$ with $v^P(S) = 1$ if $N^*_P \subseteq S$ and $v^P(S) = 0$ otherwise.
The path-counting game $v$ is a linear combination of these simple games~$v^P$.
As the Shapley value is linear \citep{Shapley1953}, it follows that the Shapley value of $v$ is obtained by summing the Shapley values of the simple games, which yields $w^*$.

\begin{theorem} \label{TH:shapley}
    The Shapley value of the path-counting game $v$ equals~$w^*$. 
\end{theorem}

If adding $i$ to coalition $S$ completes path $P$, then adding $i$ to a larger coalition $T \supseteq S$ also completes~$P$.
Hence, if $S \subseteq T$, then $v(S \cup \{i\}) - v(S) \leq v(T \cup \{i\}) - v(T)$.
That is to say, the path-counting game $v$ is convex and, therefore, the Shapley value is in the game's core \citep{Shapley1971}.
Although core stability is not central to our application, we can interpret this as that no coalition is unfairly treated:
no coalition is allocated less weight than the share of paths it realizes on its own.

\subsection{Polynomial-time computation} \label{SUB:poly}

An obstacle to using the Shapley value in practical applications is that exact computation quickly gets intractable. 
To compute $i$'s contribution to every coalition, the number of steps rapidly grows prohibitively large;
already with $n = 30$ agents, there are $2^{30} \approx 10^9$ coalitions.
In a sufficiently sparse graph with few paths, the algorithm in Subsection~\ref{SUB:weightalgo} may be tractable even if there are many nodes.
However, it too can run into a similar worst-case scenario.
The number of paths may grow exponentially with the number of agents so iterating over all paths may, again, be slow.
Figure~\ref{FIG:grid} gives an example of a grid-like graph with $2m + 2$ nodes and $2^m$ paths.
With $n \approx 60$ agents, we may have a billion paths to sum over.

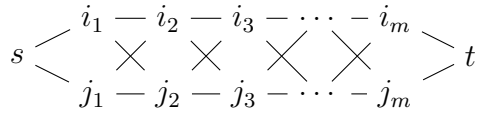
\begin{figure}[!htb]
    \centering
    \begin{tikzpicture}
        \node at (0,0) (s) {$s$};
        \node at (6,0) (t) {$t$};
        \foreach \x in {1,2,3}{
            \node at (\x, .5) (i\x) {$i_\x$};
            \node at (\x,-.5) (j\x) {$j_\x$};
        }
        \node at (4, .5) (i4) {$\dots{}$};
        \node at (4,-.5) (j4) {$\dots{}$};
        \node at (5, .5) (i5) {$i_m$};
        \node at (5,-.5) (j5) {$j_m$};
        \draw (s) -- (i1);
        \draw (s) -- (j1);
        \foreach \x/\y in {1/2,2/3,3/4,4/5}{
            \draw (j\y) -- (i\x) -- (i\y);
            \draw (j\y) -- (j\x) -- (i\y);
        }
        \draw (i5) -- (t);
        \draw (j5) -- (t);
    \end{tikzpicture}
    \caption{A graph in which the number of paths grows exponentially with the number of two-agent ``layers''.}
    \label{FIG:grid}
\end{figure}

The key to ensure scalable, fast computation is to note that agent $i$ is attributed the same amount for each five-agent path no matter who the other four agents are.
For this reason, it suffices to find the \emph{distribution} of the agent's paths (how many of length~$1$, length~$2$, and so on) and one does not need to track all paths.
This then becomes a simple exercise in dynamic programming.
We exploit the topological order to make two passes:
one forward to compute subpaths from the source, and one backwards to compute continuations to the sinks.
Intuitively, if there are two subpaths from the source to agent $i$ of length $1$ and $2$ and two subpaths of length $3$ and $4$ from $i$ to the sinks, then $i$ is part of four source-sink paths:
one of length $1 + 3$, two of length $1 + 4 = 2 + 3$, and one of length $2 + 4$.

Let $n = \lvert N \rvert$ denote the number of nodes. 
Let $F(x, j) \in \mathbb{Z}$ denote the number of subpaths of length $x \in \{0, \dots, n\}$ from the source to agent~$j$.
This is computed as follows.
Begin by setting $F(0, s) = 1$ to capture the trivial subpath from the source to itself.
Moreover, set $F(x, s) = 0$ for positive lengths $x > 0$ and $F(0,j) = 0$ for all non-source nodes $j \neq s$.
Proceed in topological order, starting with nodes that have incoming links only from the source.
For each agent~$j$ and length~$x \in \{0, \dots, n\}$, let $N^-_j \equiv \{ i \in N \mid i \to j \}$ and set 
\[
    F(x+1, j) = \sum_{i \in N^-_j} F(x, i).
\]
Intuitively, for each subpath that reaches $i$ in $x$ steps, we concatenate the edge $i \to j$ and obtain one subpath that reaches $j$ in $x+1$ steps. 
In total, $F$ is computed by evaluating roughly $n^2$ sums.

Thereafter, we do an analogous backward pass:
let $B(x, i) \in \mathbb{Z}$ denote the number of subpaths of length $x \in \{0, \dots, n\}$ from agent $i$ to some sink.
This time, for each sink~$t$, set $B(0, t) = 1$, $B(x, t) = 0$ for $x > 0$, and $B(0,i) = 0$ for non-sink nodes $i\in N^*$.
Proceed instead from the end in reverse topological order;
in general, set 
\[
    B(x+1, i) = \sum_{j \in N^+_i} B(x, j).
\]

Let $P(y, i) \in \mathbb{Z}$ denote the number of source-sink paths of length $y \in \{1, \dots, n\}$ that agent~$i$ is part of.
This can be computed through the convolution
\[
    P(y, i) = \sum_{x = 1}^y F(x,i) \cdot B(y-x,i).
\]
(Strictly speaking, there are $x$ agents before $i$ and $y-x$ agents after $i$, out of which one is a sink;
hence, once we include $i$ in the path and exclude the sink, we indeed have $y$ non-sink agents on the path.)
This is a sum of roughly $n$ multiplications.%
\footnote{For large graphs, the process may be sped up through a fast Fourier transformation \citep{CooleyTukey1965}.
As a benchmark, computing $w^*$ on a normal laptop for a graph with $n = 1\,000$ nodes and $10^{17}$ paths takes roughly one second.
This would be intractable with the procedures in Subsections~\ref{SUB:weightalgo} and~\ref{SUB:Shapley}.}
Finally, we have
\[
    w^*_i = \frac{1}{\lvert{\mathcal{P}\rvert}} \sum_{y = 1}^n P(y, i) / y.
\]

\section{Simulation study} \label{SEC:simulations}

To gain further insights on our fixed-weight solution $\phi^*$, we compare it numerically to the simple benchmark in which each agent is liable only for the loss they directly cause to their successor.
This is the natural status quo in the absence of coordination via a common contract:
\[
    \widehat\phi_i(P,\ell) =
    \begin{cases}
        \ell(ij) & \text{if } ij \in E_P  \\
        0 & \text{otherwise.}
    \end{cases}
\]

We construct a tiered supply-chain network with five production layers and a sixth consumer layer.
Links are mainly from one layer to the next, but we also allow ``shortcuts'' that skip one layer.
The network has an ``hourglass'' structure with many suppliers feeding into a smaller set of intermediaries before distributing to a larger set of consumers.
Specifically, the six-layer network has 30, 20, 15, 10, 15, and 20 nodes per layer and its edges are randomly generated.
For each pair of nodes in layers $\ell$ and $\ell + 1$, there is a link with probability $40\%$;
for nodes in layers $\ell$ and $\ell + 2$, there is a link with probability $10\%$.
Sinks are in the consumer layer only, and each sink is reachable from some source (base-layer) node.

Although the graph no longer has a unique source node, extending $\phi^*$ to this case is straightforward.
We treat each of the 30 base-layer nodes as a source, one at a time.
For a given source, we identify the subgraph reachable from this node and apply our path-counting algorithm on this subgraph.
We then draw a loss function at random (independently and uniformly from $0$ to $100$ for each edge).
We compute the efficient cost, our fixed-weight solution~$\phi^*$, and the ``local-liability'' rule $\widehat\phi$.
This is repeated for $10\,000$ randomly generated loss functions before proceeding to the next source node.
Liabilities are averaged over all $30 \cdot 10\,000$ instances.
The data is summarized in Table~\ref{TAB:sims}.

\begin{table}[!htb]
    \centering
    \begin{tabular}{ll p{1.5cm} p{1.5cm} p{1.5cm} p{1.5cm}} \toprule
        & & \multicolumn{2}{l}{Average liabilities} & \multicolumn{2}{l}{Squared liabilities} \\ 
        Layer & Agents & $\phi^*$ & $\widehat\phi$ & $\phi^*$ & $\widehat\phi$ \\ \midrule
        0 & 30 & 0.26 & 0.34 & 2.43 & 6.77 \\
        1 & 20 & 0.38 & 0.6 & 0.58 & 16.47 \\
        2 & 15 & 0.5 & 0.99 & 0.45 & 29.5 \\
        3 & 10 & 0.74 & 0.91 & 0.69 & 19.68 \\
        4 & 15 & 0.5 & 0.71 & 0.34 & 17.44 \\ \midrule
        All & 90 & 0.42 & 0.63 & 1.15 & 15.93 \\ \bottomrule
    \end{tabular}
    \caption{Summary statistics of simulated liabilities under the fixed-weight rule $\phi^*$ and the local-liability benchmark $\widehat{\phi}$.}
    \label{TAB:sims}
\end{table}

Given that the network is relatively small and the losses are randomly generated, we focus on qualitative insights.
First, as expected, the efficient liabilities under $\phi^*$ do better in terms of total losses incurred;
our results suggests that $\widehat\phi$ would increase total losses by about $50\%$.
This is partly explained by $\phi^*$ inducing slightly shorter paths (by taking more shortcuts):
the average path length is roughly $10\%$ longer under the alternative rule.
Looking at the distribution across agents, we find that $\phi^*$ generates a more equal liability profile (as measured by the Gini coefficient), consistent with the idea that fixed weights encode a form of solidarity.
If we instead look at the average squared liability, every agent is better off under our fixed-weight solution.
This points to an insurance-like aspect of the rule:
each agent always pays a little;
for contrast, liabilities exhibit much larger variation under the local-liability rule.

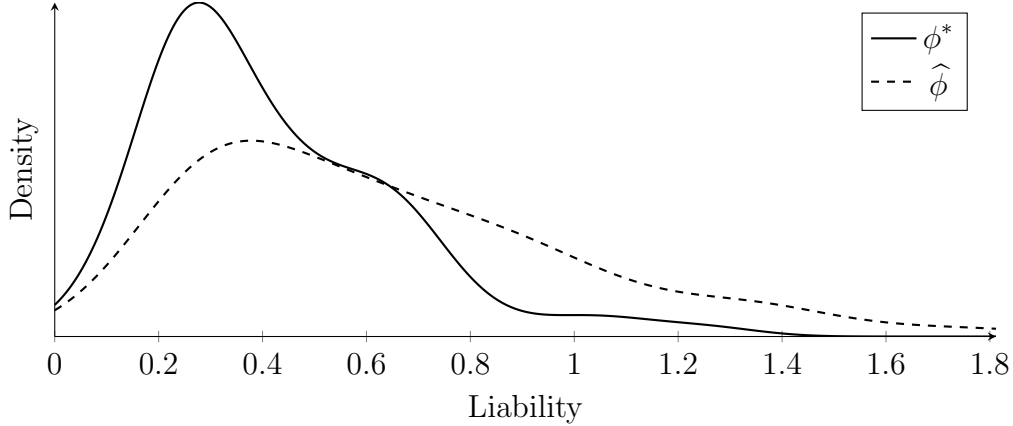
\begin{figure}
\centering
\begin{tikzpicture}
\begin{axis}[
    width=0.85\textwidth,
    height=6cm,
    xlabel={Liability},
    ylabel={Density},
    ytick=\empty,
    legend style={at={(0.97,0.97)},anchor=north east},
    axis lines=left,
]

\addplot[
    thick,
]
table[x=x,y=shapley,col sep=comma]{kde.csv};
\addlegendentry{$\phi^*$}

\addplot[
    thick,
    dashed,
]
table[x=x,y=greedy,col sep=comma]{kde.csv};
\addlegendentry{$\widehat{\phi}$}

\end{axis}
\end{tikzpicture}
\caption{Kernel density estimates of individual liabilities under the fixed-weight rule $\phi^*$ and the local-liability benchmark $\widehat{\phi}$.}
\label{FIG:kde}
\end{figure}

Still, not all agents are better off.
Out of the 90 non-sink agents, 74 are better off under~$\phi^*$.
Intuitively, a central distribution center with many outgoing edges will be part of many paths and therefore be assigned a high weight under our rule.
By contrast, under $\widehat\phi$, more outgoing edges is beneficial:
the agent only covers the loss of the canceled edge, and this can be done cheaply if there are many options.
A simple way to illustrate the difference is to fix the node's total degree (that is, sum of in- and out-degree).
Under $\widehat\phi$, liability increases in in-degree and decreases in out-degree;
under our fixed-weight solution $\phi^*$, liability instead increases most strongly in the \emph{product} of in- and out-degree.

\section{Discussion} \label{SEC:discussion}

We briefly sketch two natural generalizations. 
In both cases, the main insights from our analysis remain intact under the broader modeling framework.

\subsection{Incomplete information on losses}

In our model, the loss function is common knowledge.
In practice, this assumption may be restrictive.
For instance, agent relationships may be governed by bilateral contracts under which each agent knows only the losses associated with their own contracts (edges).
Nevertheless, our results are robust under this more realistic specification of incomplete information.
To see this, suppose that the loss on edge $ij$ is private information known only to agents~$i$ and~$j$.
Consider the following two-stage procedure.
First, all agents report losses for the edges they know.
Subsequently, agents make cancellation choices as in Section~\ref{SEC:model} but with information on the reported losses instead.
Under a fixed-weight rule, reports will be truthful:
no agent can gain from misreporting edge losses, as doing so can only distort others' efficient decisions and increase total losses and, in turn, the agent's own liability.

In similar spirit to how smart contracts can implement the liability rule, they can be extended to make this two-stage procedure operational as well.
That is, the contract would act as a coordination device that collects reports, registers agent actions, and assigns liability;
see, for instance, \citet{GudmundssonHougaardKo2023} for a sketch of such a setup.

\subsection{Generalized cancellation games}

Fixed-weight rules are very robust and ensure efficient implementation in general.
To illustrate, consider the following sequential-move game.
We maintain the assumption of a directed acyclic graph $(N,E)$ with a unique source $s$ but allow for more general forms of cancellation structures:
\begin{enumerate}
    \item The source $s$ chooses an action $a_s$ from a given set of actions $A_s$.
    So far, $A_s$ has been the agents adjacent to the source;
    now, it can be far more general.
    For instance, an action could represent a subset of adjacent agents (in case $s$ has to cancel multiple edges) or combine an edge with a ``degree of cancellation'' (e.g., capturing a partial breakdown), and more.
    The choice $a_s$ determines the next agent, $i_1 = f(a_s) > s$ (hence, even if the choice represents multiple affected agents, we still proceed sequentially).
    \item Agent $i_1$ chooses an action $a_1$ from some action set $A_1(a_s)$.
    The available actions may depend on the source's choice:
    for instance, if $a_s, a'_s \in A_s$ capture different cancellations decisions, then we may have $A_1(a_s) \neq A_1(a'_s)$ even if $f(a_s) = f(a'_s)$.
    Jointly, the choices $a_s, a_1$ determine the next agent, $i_2 = f(a_s, a_1) > i_1$.
    \item [$\vdots$]
    \item [$k$.] In general, agent $i_k$ chooses an action $a_k$ from some action set $A_k(a_s, \dots, a_{k-1})$.
    All choices up to this point determine the next agent, $i_{k+1} = f(a_s, a_1, \dots, a_k) > i_k$.
\end{enumerate}
As there is a finite number of agents, the game ends in a finite number of steps.
The non-negative total losses generated through this realization of the game is a function of all choices $a_s, a_1, \dots, a_t$.
As long as these losses are shared according to a fixed-weight rule, equilibrium choices will be efficient (redefined to fit the new game structure).

\section{Concluding remarks} \label{SEC:concluding}

We have studied liability assignment on networks where agents' actions jointly determine both efficiency and loss allocation.
Our axiomatic analysis has singled out the class of fixed-weight rules, which achieve efficient implementation using balanced, scale-invariant transfers that ensure consistent liability across efficient outcomes.
We propose a particular fixed-weight rule, justified in part by its connection to cooperative game theory.
Specifically, the associated weights result from application of the Shapley value to a related cooperative game.
In this regard, we have emphasized the computational features that are specific to our particular model.
Finally, we have argued that fixed-weight rules in general are robust;
even in broader domains, they can be used to implement efficient outcomes via simple coordination devices. 
This suggests that fixed-weight rules provide a tractable and transparent approach to allocating responsibility for cascading network failures.

\bibliography{bibliography}
\bibliographystyle{abbrvnat}

\begin{appendix}

\section{Proofs} \label{APP:proofs}
\setcounter{proposition}{0}
\setcounter{theorem}{0}
\setcounter{observation}{0}

\begin{proposition}[Downstream monotonicity] 
    Let $\phi$ satisfy \textit{realized-loss dependence} and \textit{efficient implementation}.
    Consider an arbitrary loss function $\ell\in\mathcal{L}$, agent $i \in N$, and subpath $P \in \mathcal{P}^-_i$.
    For all agents $j, k \in N^+_i$ and efficient continuation subpaths $P_j\in \mathcal{E}_j(\ell)$ and $P_k\in \mathcal{E}_k(\ell)$, 
    \[
        \ell(ij) + L_j < \ell(ik) + L_k 
        \, \iff \,
        \phi_{i}(P \oplus ij \oplus P_j,\ell) < \phi_{i}(P \oplus ik \oplus P_k,\ell).
    \]
\end{proposition}

\begin{proof}
    Let $\phi$ satisfy \textit{efficient implementation} and \textit{realized-loss dependence}.
	Fix a loss function~$\ell$, a history ${P} \in \mathcal{P}^-_i$ ending with agent~$i\in N$, two $i$-successors $j,k\in N_i^+$, and two (sub-) paths $P_j\in\mathcal{E}_j(\ell)$ and $P_k \in \mathcal{E}_{k}(\ell)$.
	Let $\olsi{P}_j \equiv P \oplus ij \oplus P_j\in\mathcal{P}$ and $\olsi{P}_k \equiv P \oplus ik \oplus P_k\in\mathcal{P}$.
	We wish to show that
	\[
        \ell(ij)+L_j < \ell(ik)+L_k \iff \phi_{i}(\olsi{P}_j,\ell) < \phi_{i}(\olsi{P}_k,\ell).
    \]
    Begin by defining the set of all paths that can be formed using edges only from $\olsi{P}_j$ and $\olsi{P}_k$:
    \[
    \mathcal{O}(\olsi{P}_j,\olsi{P}_k) \equiv \left\{\wtsi{P}\in\mathcal{P}: E_{\wtsi{P}}\subset (E_{\olsi{P}_j}\cup E_{\olsi{P}_k}) \right\}.
    \]
    Now, notice that any path $\wtsi{P}\in \mathcal{O}(\olsi{P}_j,\olsi{P}_k)$ necessarily begins with $P$ and is immediately followed by either $ij$ or $ik$.
    Therefore, the lowest loss of any such path is
    \[
        \min_{\wtsi{P}\in\mathcal{O}(\olsi{P}_j,\olsi{P}_k)} \ell(\wtsi{P})=
        \ell(P)+\min \left\{ \ell(ij)+L_j,\ell(ik)+L_k \right\} = 
        \min\{\ell(\olsi{P}_j),\ell(\olsi{P}_k)\}.
    \]
    So, for all paths $\wtsi{P}\in\mathcal{O}(\olsi{P}_j,\olsi{P}_k)$ we have that $\ell(\wtsi{P})\geq 
    \min\{\ell(\olsi{P}_j),\ell(\olsi{P}_k)\}$.
    Define also a new loss function $\tilde\ell \in \mathcal{L}$ through 
    \[
        \tilde\ell(e) 
        \equiv
            \begin{cases}
                \ell(e) & \text{if }  e \in  \olsi{P}_j\cup\olsi{P}_k\\
                \ell(\olsi{P}_j) + 1 & \text{otherwise.}
	      \end{cases}
    \]

    ``$\Rightarrow$'':
 Assume that $\ell(ij)+L_j<\ell(ik)+L_k$ or, equivalently, that $\ell(\olsi{P}_j) < \ell(\olsi{P}_k)$.
    From the construction of $\tilde\ell$, this assumption means that $ \tilde \ell(\olsi{P}_j)\leq\tilde \ell(\wtsi{P})$ for every $\wtsi{P}\in \mathcal{P}$ and thus $\olsi{P}_j \in \mathcal{E}(\tilde\ell)$, whereas $ \olsi{P}_k \notin \mathcal{E}(\tilde\ell)$.
    By \textit{efficient implementation}, $\olsi{P}_j \in \mathcal{S}(\tilde\ell;\phi)$ and $\olsi{P}_k \notin \mathcal{S}(\tilde\ell;\phi)$, which jointly imply $\phi_i(\olsi{P}_j, \tilde\ell) < \phi_i(\olsi{P}_k ,\tilde\ell)$.
    By \textit{realized-loss dependence}, $\phi_i(\olsi{P}_j, \tilde\ell) = \phi_i(\olsi{P}_j, \ell)$ and $\phi_i(\olsi{P}_k, \tilde\ell) = \phi_i(\olsi{P}_k, \ell)$. 
    Hence, $\phi_i(\olsi{P}_j, \ell) < \phi_i(\olsi{P}_k, \ell)$.

    ``$\Leftarrow$'':
  Assume that $\phi_i(\olsi{P}_j, \ell)<\phi_i(\olsi{P}_k, \ell)$ and consider the game $\mathcal{G}(\tilde\ell;\phi)$.
    By way of contradiction, say that $\ell(ij) +L_j \geq \ell(ik) + L_k$ or, equivalently, that $\ell(\olsi{P}_j) \geq \ell(\olsi{P}_k)$.
  Then, from the construction of $\tilde{\ell}$, $\tilde{\ell}(\olsi{P}_j) \geq \tilde\ell(\olsi{P}_k)$.
  This implies $\olsi{P}_k \in \mathcal{E}(\tilde\ell)$, since for every path $\wtsi{P} \in \mathcal{P}$, we have $\tilde\ell(\wtsi{P}) \geq \tilde\ell(\olsi{P}_k)$.
    Then, by \textit{efficient implementation}, we have $\olsi{P}_k \in \mathcal{S}(\tilde\ell;\phi)$, which implies $\phi_i(\olsi{P}_j, \tilde\ell) \geq \phi_i(\olsi{P}_k, \tilde\ell)$.
    Hence, by \textit{realized-loss dependence}, $\phi_i(\olsi{P}_j, \ell) \geq \phi_i(\olsi{P}_k, \ell)$, contradicting our assumption that $\phi_i(\olsi{P}_j,\ell) < \phi_i(\olsi{P}_k, \ell)$.
    So, it has to be that $\ell(ij) +L_j < \ell(ik) + L_k$.
\end{proof}

\begin{proposition}[Efficient-path invariance] 
    Let $\phi$ satisfy \textit{efficient implementation} and \textit{pairwise collusion-proofness}.
    For each loss function $\ell \in \mathcal{L}$ and all efficient paths $P,P' \in \mathcal{E}(\ell)$, 
    \[
        \phi(P,\ell) = \phi(P',\ell).
    \]
\end{proposition}

\begin{proof}
    Consider two efficient paths $P,P' \in \mathcal{P}$ and let agent $i \in N$ be the first to make a different choice in $P$ and~$P'$.
    Then $i$ must be assigned the same liability under the two:
    if not, then $i$ would prefer one to the other, say $P$ to $P'$, and $P'$ would not be an equilibrium path.
    This would contradict \textit{efficient implementation}.
    Hence, $\phi_i(P,\ell) = \phi_i(P',\ell)$.
    By \textit{pairwise collusion-proofness} applied from $P$ to $P'$ and from $P'$ to $P$, we have $\phi_j(P,\ell) = \phi_j(P',\ell)$ for all agents $j \neq i$.
    Hence, $\phi(P,\ell) = \phi(P',\ell)$.
\end{proof}

Lemma~\ref{LE:irreducible} will be used repeatedly in later proofs.
It asserts that, fixing the losses on a particular path, we can construct a new loss function under which all paths are efficient.

\begin{lemma}[Irreducible losses given path] \label{LE:irreducible}
    For each path $P \in \mathcal{P}$ and loss function $\ell \in \mathcal{L}$, there exists a loss function $\hat\ell \in \mathcal{L}$ with $\mathcal{E}(\hat\ell) = \mathcal{P}$ and $\hat\ell(e) = \ell(e)$ for all edges $e \in E_P$.
\end{lemma}

\begin{proof}
    Fix path $P \in \mathcal{P}$ and loss function $\ell \in \mathcal{L}$.
    We construct an auxiliary variable $c \in \mathbb{R}^N$ as follows.
    For all sink agents $t \in N$, set $c_t = \ell(P)$.
    For the source, set $c_s = 0$.
    For each interior on-path agent $i \in N_P$, let $P^-_i \in \mathcal{P}^-_i$ denote the subpath $s \to \dots \to i$ of $P$ leading up to~$i$ and set $c_i = \ell(P^-_i)$.
    Finally, process off-path non-sink agents in topological order:
    for each agent $j \not \in N_P$, set $c_j = \max_{i \to j} c_i$.
    
    Construct loss function $\hat\ell \in \mathcal{L}$ such that $\hat\ell(ij) = c_j - c_i \geq 0$ for each edge $i \to j$.
    Hence, $\hat\ell$ matches $\ell$ on~$P$.
    Take an arbitrary path $\tilde{P} \in \mathcal{P}$ and label it $\tilde{P} = (s, 1, \dots, m, t)$.
    Then
    \[
        \hat\ell(\tilde{P}) 
        = \hat\ell(s,1) + \hat\ell(1,2) + \dots + \hat\ell(m,t)
        = (c_1 - 0) + (c_2 - c_1) + \dots + (c_{\tilde{t}} - c_m))
        = c_t
        = \ell(P).
    \]
    That is, every path has the same loss under $\hat\ell$.
    Thus, $\mathcal{E}(\hat\ell) = \mathcal{P}$.
\end{proof}

\begin{proposition}[Redistribution invariance] 
    Let $\phi$ satisfy \textit{efficient implementation}, \textit{realized-loss dependence}, and \textit{pairwise collusion-proofness}.
    For each path $P \in \mathcal{P}$ and all loss functions~${\ell,\ell' \in \mathcal{L}}$, 
    \[
        \ell(P) = \ell'(P)
        \, \implies \,
        \phi(P,\ell) = \phi(P,\ell').
    \]
\end{proposition}

\begin{proof}
    Let $\phi$ satisfy \textit{efficient implementation}, \textit{realized-loss dependence}, and \textit{pairwise collusion-proofness}.
    By Proposition~\ref{PR:EPI}, $\phi$ satisfies \textit{efficient-path invariance}.
    Fix a path $P \in \mathcal{P}$ and a loss function $\ell \in \mathcal{L}$.
    We wish to show that, for any loss function $\ell' \in \mathcal{L}$ that keeps total losses on~$P$ unchanged, so $\ell(P) = \ell'(P)$, we have $\phi(P,\ell) = \phi(P,\ell')$.
    To arrive at this point, we will move from $\ell$ to $\ell'$ in small steps:
    specifically, we will repeatedly shift losses only between two adjacent edges on path $P$.
    In \textsc{Part~I}, we show that this pairwise adjacent shift does not affect liabilities, going through a sequence of loss functions that keep $\ell(P)$ unchanged.
    In \textsc{Part~II} and~\textsc{III}, this conclusion is extended to any loss redistribution along path $P$.

    \vspace{.25\baselineskip}
    \textsc{Part I:}
    Consider two edges $ij, jk \in E_P$.
    Define loss function $\ell' \in \mathcal{L}$ by shifting some losses from $ij$ onto $jk$.
    Specifically, let $0 \leq \delta \leq \ell(ij)$ and set $\ell'(ij) = \ell(ij) - \delta$, $\ell'(jk) = \ell(jk) + \delta$, and $\ell'(e) = \ell(e)$ for $e \neq ij, jk$.
    We proceed by considering a sequence of loss functions that eventually gets us to $\ell'$.
    
    By Lemma~\ref{LE:irreducible}, there is a loss function $\ell^1 \in \mathcal{L}$ that matches $\ell$ on~$P$, so $\ell^1(e) = \ell(e)$ for all $e \in E_P$, such that $\mathcal{E}(\ell^1) = \mathcal{P}$.
    By \textit{realized-loss dependence}, $\phi(P,\ell) = \phi(P,\ell^1)$.
    By Assumption~\ref{A2}, there is a path $P' \in \mathcal{P}$ that does not go through agent~$j$;
    that is, $j\notin N_{P'}$.
    As $P,P' \in \mathcal{E}(\ell^1)$, by \textit{efficient-path invariance}, $\phi(P,\ell^1) = \phi(P',\ell^1)$.

    We define a new loss function $\ell^2 \in \mathcal{L}$ that matches $\ell^1$ everywhere except for the edges that involve~$j$.
    Specifically, the loss on the edge $ij$ decreases by $\delta$ whereas those of all $j$'s outgoing edges increase by the same amount:
    \[
        \ell^2(e) = \begin{cases}
            \ell^1(e) - \delta & \text{if $e = ij$,} \\
            \ell^1(e) + \delta & \text{if $e\in E_j$,} \\
            \ell^1(e) & \text{otherwise.} \\
        \end{cases} 
    \]
    As $P'$ does not go through~$j$, its losses are unchanged.
    By \textit{realized-loss dependence}, $\phi(P',\ell^1) = \phi(P',\ell^2)$. 
    Moreover, for each path $\tilde{P} \in \mathcal{P}$, $\ell^2(\tilde{P}) \geq \ell^1(\tilde{P})$:
    the only edge that is cheaper at $\ell^2$ than at $\ell^1$ is $ij$, but all paths that include $ij$ must also include one of $j$'s outgoing edges, which has gotten equally more expensive.
    The construction of $\ell^2$ implies $\ell^2(P) = \ell^1(P)$ and $\ell^2(P') = \ell^1(P')$, and since $P,P' \in \mathcal{E}(\ell^1)$, we have $P,P' \in \mathcal{E}(\ell^2)$.
    By \textit{efficient-path invariance}, $\phi(P',\ell^2) = \phi(P,\ell^2)$.

    Finally, we recover loss function $\ell'$ as $\ell'(e) = \ell^2(e)$ for $e \in E_P$ and $\ell'(e) = \ell(e)$ for $e \not \in E_P$.
    By \textit{realized-loss dependence}, $\phi(P,\ell^2) = \phi(P,\ell')$.
    In conclusion, $\ell(P) = \ell'(P)$ and
    \[
        \phi(P,\ell)
        = \phi(P,\ell^1)
        = \phi(P',\ell^1)
        = \phi(P',\ell^2)
        = \phi(P,\ell^2)
        = \phi(P,\ell').
    \]

    Here, we obtained $\ell'$ from $\ell$ by shifting losses from $ij$ onto $jk$.
    This direction is without loss of generality:
    we could also have started from $\ell'$, shifted losses onto $ij$ to get to $\ell$, and again concluded that $\phi(P,\ell) = \phi(P,\ell')$.

    \vspace{.25\baselineskip}
    \textsc{Part II:}
    Let loss function $\ell' \in \mathcal{L}$ be such that $\ell'(P) = \ell(P)$ and $\ell'(e) = \ell(e)$ for $e \not \in E_P$.
    That is, consider any redistribution on $P$, not limited to only pairs of adjacent edges. 
    We can define a sequence of loss functions $\ell = \ell_0, \ell_1, \dots, \ell_m = \ell'$ such that $\ell_{k+1} \in \mathcal{L}$ is obtained from $\ell_k \in \mathcal{L}$ through a pairwise adjacent on-path loss shift. 
    By repeatedly applying the conclusion from \textsc{Part~I}, $\phi(P,\ell) = \phi(P, \ell_0) = \dots = \phi(P, \ell_m) = \phi(P, \ell')$.

    \vspace{.25\baselineskip}
    \textsc{Part III:}
    Finally, let loss function $\ell' \in \mathcal{L}$ be such that $\ell'(P) = \ell(P)$. 
    Let loss function $\tilde\ell \in \mathcal{L}$ be such that $\tilde\ell(e) = \ell(e)$ for $e \not \in E_P$ and $\tilde\ell(e) = \ell'(e)$ for $e \in E_P$.
    By \textit{realized-loss dependence}, $\phi(P,\ell') = \phi(P,\tilde\ell)$.
    By \textsc{Part~II}, $\phi(P,\tilde\ell) = \phi(P,\ell)$.
    Hence, $\phi(P,\ell) = \phi(P,\ell')$.
\end{proof}

\begin{proposition}[Path independence] 
    Let~$\phi$ satisfy \textit{efficient implementation}, \textit{realized-loss dependence}, and \textit{pairwise collusion-proofness}.
    For all paths~$P,P' \in \mathcal{P}$ and each loss function~${\ell \in \mathcal{L}}$, 
    \[
        \ell(P) = \ell(P') 
        \, \implies \,
        \phi(P,\ell) = \phi(P',\ell).
    \]
\end{proposition}

\begin{proof}
    Fix loss function $\ell \in \mathcal{L}$ and paths $P,P' \in \mathcal{P}$ such that $\ell(P) = \ell(P')$.
    By Lemma~\ref{LE:irreducible}, there is a loss function $\hat\ell \in \mathcal{L}$ that matches $\ell$ on $P$, so $\hat\ell(e) = \ell(e)$ for all edges $e \in E_P$, such that $\mathcal{E}(\hat\ell) = \mathcal{P}$.
    In particular, $\hat\ell(P) = \hat\ell(P')$.
    By \textit{realized-loss dependence}, $\phi(P,\ell) = \phi(P,\hat\ell)$.
    Moreover, $\hat\ell(P') = \hat\ell(P) = \ell(P) = \ell(P')$.
    That is, $\hat\ell$ is a redistribution along~$P'$;
    by Proposition~\ref{PR:redistribution}, $\phi(P',\ell) = \phi(P',\hat\ell)$.
    By Proposition~\ref{PR:EPI}, $\phi$ satisfies \textit{efficient-path invariance}, so $\phi(P,\hat\ell) = \phi(P',\hat\ell)$.
    Thus, $\phi(P,\ell) = \phi(P',\ell)$.
\end{proof}

\begin{theorem}[Characterization of fixed-weight rules] 
    A rule~$\phi$ satisfies \textit{efficient implementation}, \textit{realized-loss dependence}, \textit{paiwise collusion-proofness}, and \textit{scale invariance} if and only if $\phi = \phi^w$ for some $w \in \Delta^N_\star$.
\end{theorem}

\begin{proof}
    It is immediate that all fixed-weight rules satisfy the axioms, so we focus on the proof's other direction.
    
    Let $\phi$ satisfy \textit{efficient implementation}, \textit{realized-loss dependence}, \textit{pairwise collusion-proofness}, and \textit{scale invariance}.
    Fix a path $P^*\in\mathcal{P}$ and loss function $\ell^*\in\mathcal{L}$ with $\ell^*(P^*)>0$.
    Define weights $w$ such that, for each agent $i \in N$,
    \[
        w_i \equiv \frac{\phi_i(P^*,\ell^*)}{\ell^*(P^*)} \geq 0.
    \]
    This is non-negative as losses and liabilities are non-negative.
    By balance, $\sum_j \phi_j(P,\ell) = \ell(P)$ implies that $\sum_j w_j = 1$.
    Hence, $w \in \Delta^N$.
    
    Now, take an arbitrary path $P \in \mathcal{P}$ and loss function $\ell \in \mathcal{L}$.
    Let $\alpha \equiv \ell(P) / \ell^*(P^*) \geq 0$.
    By Corollary~\ref{CO:total}, as $\ell(P) = \alpha \cdot \ell^*(P^*)$, we have $\phi(P,\ell) = \phi(P^*, \alpha \cdot \ell^*)$.
    By \textit{scale invariance}, 
    \[
        \phi(P^*,\alpha\cdot\ell^*)
        = \alpha\cdot \phi(P^*,\ell^*)
        = \frac{\ell(P)}{\ell^*(P^*)} \cdot \phi(P^*,\ell^*) 
        = w\cdot \ell(P).
    \]
    That is, $\phi(P,\ell) = w \cdot \ell(P)$, so $\phi=\phi^w$.

    To complete the proof, we show that agents with multiple outgoing edges have positive weights.
    For contradiction, suppose not;
    say $w_i = 0$ for an agent $i \in N$ with multiple outgoing edges.
    Fix path $\tilde P \in \mathcal{P}$ and loss function $\tilde\ell \in \mathcal{L}$.
    As in Lemma~\ref{LE:irreducible}, construct a loss function $\hat\ell$ to match $\tilde\ell$ on $\tilde P$ such that $\mathcal{E}(\hat\ell) = \mathcal{\tilde P}$.
    Thereafter, increase the losses on all but one of $i$'s outgoing edges.
    Then $i$ has a unique efficient choice, yet with $w_i = 0$, $i$ is indifferent between all choices. 
    Hence, there would be inefficient equilibria, contradicting \textit{efficient implementation}.
    Therefore, $w_i > 0$ and we conclude that $w \in \Delta^N_\star$.
\end{proof}

\begin{theorem} 
    The Shapley value of the path-counting game $v$ equals~$w^*$. 
\end{theorem}

\begin{proof}
    Let $\varphi^* \in \mathbb{R}^N$ denote the Shapley value of the path-counting game~$v$.
    We use the random-permutation characterization of the Shapley value \citep[e.g.][]{Roth1988}.
    For a uniformly drawn permutation $\pi$ of $N^*$, the Shapley value of agent $i$ is
    \[
        \varphi^*_i =
            \mathbb{E}_\pi
            \left[
                v \left( \Pred_\pi(i) \cup \{i\} \right)
                -
                v \left (\Pred_\pi(i) \right)
            \right],
    \]
    where $\Pred_\pi(i) \subseteq N^* \setminus \{ i \}$ denotes the set of agents that appear before $i$ in the permutation~$\pi$.
    Next, take a path $P \in \mathcal{P}$.
    
    If $i \in N^*_P$, then $P$ contributes $1 / \lvert \mathcal{P} \rvert$ to $v(S \cup \{i\}) - v(S)$ when $N^*_P \setminus \{i\} \subseteq S = \Pred_\pi(i)$.
    In words, $i$ is the one who ``completes'' the path~$P$.
    This is when $i$ is last among the $\lvert N^*_P \rvert$ agents of $N^*_P$ in the random order $\pi$, which occurs with probability $1 / \lvert N^*_P \rvert$. 
    Thus, the expected contribution of $P$ to $\varphi^*_i$ is $1 / (\lvert N^*_P \rvert \cdot \lvert \mathcal{P} \rvert)$, which coincides with $v_i(P)$ in Subsection~\ref{SUB:weightalgo}.
    
    If $i \not \in N_P$, then $P$ never contributes to $v(S \cup \{i\}) - v(S)$, so the expected contribution is zero.
    This, again, coincides with $v_i(P)$.
    
    Summing over all paths, we have $\varphi^*_i = \sum_{P \in \mathcal{P}} v_i(P) = w^*_i$.
\end{proof}

\section{Independence of axioms} \label{APP:independence}
    
We present a series of rules below, where each satisfies all but one axiom of those imposed in Theorem~\ref{TH:fixedweight}.
Throughout, let $i$, $P$, and $\ell$ denote generic agents, paths, and losses.

\begin{itemize}
    \item Without \textit{efficient implementation}:
    Assign all liability to the source:
    $\phi^1_s(P,\ell) = \ell(P)$ and $\phi^1_i(P,\ell) = 0$ for each agent $i \neq s$.
    As all later agents will be indifferent no matter the path selected, there will be equilibrium paths that are inefficient.
    \item Without \textit{realized-loss dependence}:
    Set weights proportional to the maximum loss of outgoing edges (adding $1$ to ensure positive weights).
    For each agent $i$, let 
    \[ \textstyle
        m_i = 1 + \max_{j \in N^+_i} \ell(ij).
    \]
    Finally, set $w_i = m_i / \sum_j m_j$, and $\phi^2_i(P,\ell) = w_i \cdot \ell(P)$.
    \item Without \textit{pairwise collusion-proofness}:
    Let $m \equiv \max_\mathcal{P} \left\vert{N_P}\right\vert \geq 2$ be the length of the longest path by node count. 
    Each on-path agent gets share $\alpha = 1/m$ whereas the remaining losses are shared equally off-path.
    That is, $\phi^3_i(P,\ell) = \alpha \cdot \ell(P)$ for $i \in N_P$ and $\phi^3_i(P,\ell) = (1- \left\vert{N_P}\right\vert \alpha) / (n - \left\vert{N_P}\right\vert) \cdot \ell(P)$ for $i \not \in N_P$.
    \item Without \textit{scale invariance}:
    Use a weight function $w \colon \mathbb{R}_{\geq 0} \to \Delta^N_*$ that sets weights $w(T)$ depending on total losses $T = \ell(P)$.
    For instance, let $w_s(T) = 1 / \sqrt{T + 1}$ and $w_j(T) = (1 - w_s(T)) / (n-1)$ for agents $j \neq s$.
    Set $\phi^5_i(\ell,P) = w_i(\ell(P)) \cdot \ell(P)$.
    
    Note that $\phi^5_s(\ell,P) = \ell(P) / \sqrt{\ell(P) + 1}$ is increasing in $\ell(P)$;
    that is, even though the source's weight decreases with a larger loss total, the source's liability increases.
    Hence, all agents prefer cheaper paths, so \textit{efficient implementation} is still satisfied.
\end{itemize}

\section{Additional results} \label{APP:additional}

    \begin{observation}[Impossibility of ``on-path only'' rules] \label{OBS:pathonly}
       There are graphs for which no rule satisfies \textit{efficient implementation}, \textit{realized-loss dependence}, and, for each path $P \in \mathcal{P}$, loss function $\ell \in \mathcal{L}$, and agent~$i \in N$,
        \[
            \phi_i(P, \ell) > 0 
            \, \implies \,
            i \in N_P.
        \]
    \end{observation}

    \begin{proof}
        Consider the graph in Figure~\ref{FIG:only_on_the_path}.
        Denote its paths $P = (s \to t)$, $P' = (s \to i \to t)$, and $P'' = (s \to i \to j \to t)$.
        First, let the loss function $\ell \in \mathcal{L}$ be such that $\ell(st) = \ell(it) = \ell(jt) = 1$ and $\ell(si) = \ell(ij) = 0$.
        Hence, all paths are efficient.
        As the source is the only non-sink agent on~$P$, by balance and \textit{path-only liabilities}, $\phi_s(P, \ell) = 1$.
        By \textit{efficient implementation}, the source must be indifferent between choosing~$i$ or~$t$. 
        That is, the source’s liability should be the same for $P$ and $P'$:
        $\phi_s(P', \ell) = \phi_s(P, \ell) = 1$.
        But then, by balance, $\phi_i(P', \ell) = 0$.
        Similarly, agent $i$ must be indifferent between $j$ and~$t$, so $\phi_i(P'', \ell) = \phi_i(P', \ell) = 0$.
    
        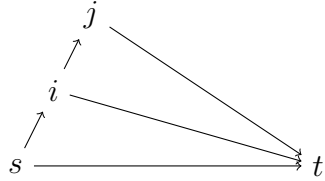
\begin{figure}[!htb]
        \centering
        \begin{tikzpicture}
            \node at (0,0) (s) {$s$};
            \node at (0.5,1) (i) {$i$};
            \node at (1,2) (j) {$j$};
            \node at (4,0) (t) {$t$};
            \draw [->] (s) -- (t); %
            \draw [->] (s) -- (i); %
            \draw [->] (i) -- (j); %
            \draw [->] (i) -- (t); %
            \draw [->] (j) -- (t); %
        \end{tikzpicture}
        \caption{Graph for Observation~\ref{OBS:pathonly}.}
        \label{FIG:only_on_the_path}
    \end{figure}

    Next, define the loss function $\ell' \in \mathcal{L}$ that matches $\ell$ except for $\ell'(it) = 0$.
    Under $\ell'$, path $P'$ is the unique efficient path.
    Yet, as losses on $P''$ are unchanged, by \textit{realized-loss dependence}, $\phi_i(P'', \ell') = \phi_i(P'', \ell) = 0 \leq \phi_i(P', \ell')$.
    That is to say, $i$ does not strictly prefer the efficient $P'$ to the inefficient $P''$, contradicting \textit{efficient implementation} at~$\ell'$.
    \end{proof}

\end{appendix}

\end{document}